\documentclass[11pt,reqno]{amsart}

\usepackage{amsmath, amsthm, amstext, amssymb, amsfonts, graphicx, bm, subfigure}
\newtheorem{theorem}{Theorem}[section]
\newtheorem{corollary}[theorem]{Corollary}
\newtheorem{definition}[theorem]{Definition}

\newtheorem{proposition}[theorem]{Proposition}
\newtheorem*{remark}{Remark}
\newtheorem*{acknowledgement}{Acknowledgement}

\newcommand{\sign}{\mathrm{sign}}
\newcommand{\x}{\mathbf{x}}
\newcommand{\y}{\mathbf{y}}
\newcommand{\z}{\mathbf{z}}
\newcommand{\m}{\mathbf{m}}
\newcommand{\n}{\mathbf{n}}
\renewcommand{\L}{\mathcal{L}}
\newcommand{\A}{\mathrm{A}}
\newcommand{\I}{\mathrm{I}}

\begin{document}

\title{{Spectral properties of the renormalization group at infinite temperature}}

\author{Mei Yin}
\address{Department of Mathematics, University of Arizona, Tucson,
AZ 85721, USA} \email{myin@math.arizona.edu}

\dedicatory{\rm July 24, 2010}

\begin{abstract}
The renormalization group (RG) approach is largely responsible for
the considerable success that has been achieved in developing a
quantitative theory of phase transitions.  Physical properties
emerge from spectral properties of the linearization of the RG map
at a fixed point. This article considers RG for classical
Ising-type lattice systems. The linearization acts on an
infinite-dimensional Banach space of interactions. At a trivial
fixed point (zero interaction), the spectral properties of the RG
linearization can be worked out explicitly, without any
approximation. The results are for the RG maps corresponding to
decimation and majority rule. They indicate spectrum of an unusual
kind: dense point spectrum for which the adjoint operators have no
point spectrum at all, only residual spectrum. This may serve as a
lesson in what one might expect in more general situations.

\keywords{Ising model; renormalization group; trivial fixed point;
residual spectrum}
\end{abstract}

\maketitle

\section{Introduction}
\label{intro} We consider renormalization group (RG)
transformations for Ising-type lattice spin systems on
$\mathbb{Z}^d$ . Our original lattice is denoted by $\L$ and our
image lattice is denoted by $\L'$. The image lattice $\L'$ indexes
a partition of $\L$ into cubical blocks, all with the same
cardinality $b^d$. Thus for each site $\y$ in $\L'$, there is a
corresponding block $\y^o$ that is a subset of $\L$, given by
\begin{equation}
\y^o=\{\x: by_i-\frac{b-1}{2}\leq x_i\leq by_i+\frac{b-1}{2},
1\leq i\leq d\}
\end{equation}
for odd blocking factor $b$; and
\begin{equation}
\y^o=\{\x: by_i-\frac{b-2}{2}\leq x_i\leq by_i+\frac{b}{2}, 1\leq
i\leq d\}
\end{equation}
for even blocking factor $b$. More generally, for each subset $Y$
of $\L'$, there is a corresponding union of blocks $Y^o$ that is a
subset of $\L$. A spin variable $\sigma_{\x}=\pm 1$ is assigned to
each site $\x$ in $\L$, and a block spin variable
$\sigma'_{\y}=\pm 1$ is assigned to each site $\y$ in $\L'$. If
$X$ is a finite subset of the original lattice, then $\sigma_X$
denotes the spin variable $\prod_{\x \in X}\sigma_{\x}$.
Similarly, if $Z$ is a finite subset of the image lattice, then
$\sigma'_Z$ denotes the block spin variable $\prod_{\z\in
Z}\sigma'_{\z}$. The main physical properties of $\L$ are encoded
in the Hamiltonian $H(\sigma)=-\sum_X J(X)\sigma_X$, where $J$ is
the original interaction defined on nonempty finite subsets of
$\L$. Likewise, the main physical properties of $\L'$ are encoded
in the Hamiltonian $H'(\sigma')=-\sum_Y J'(Y)\sigma'_Y$, where
$J'$ is the resulting interaction defined on nonempty finite
subsets of $\L'$.

Here is the formal definition of the RG map:
\begin{equation}
\frac{e^{\sum_Y J'(Y)\sigma'_Y}}{\sum_{\sigma'}e^{\sum_Y
J'(Y)\sigma'_Y}}=\frac{\sum_{\sigma}\prod_{\y \in
\L'}T_{\y}(\sigma,\sigma'_{\y})e^{\sum_X
J(X)\sigma_X}}{\sum_{\sigma}e^{\sum_X J(X)\sigma_X}},
\end{equation}
where $\sum_{\sigma}$ and $\sum_{\sigma'}$ (normalized sums)
denote the product probability measures on $\{+1,-1\}^{\L}$ and
$\{+1,-1\}^{\L'}$, respectively, and $T_{\y}(\sigma,
\sigma'_{\y})$ denotes a specific RG probability kernel, which
depends only on $\sigma$ through $\y^o$, and satisfies both a
symmetry condition,
\begin{equation}
\label{sym} T_{\y}(\sigma, \sigma'_{\y})=T_{\y}(-\sigma,
-\sigma'_{\y}),
\end{equation}
and a normalization condition,
\begin{equation}
\label{norm} \sum_{\sigma'}T_{\y}(\sigma, \sigma'_{\y})=1
\end{equation}
for every $\sigma$ and every $\y$. Notice that because of
(\ref{sym}) and (\ref{norm}),
\begin{equation}
\label{begin} \sum_{\sigma}T_{\y}(\sigma,
+1)=\sum_{\sigma}T_{\y}(\sigma, -1)=1.
\end{equation}
In the following, we restrict our attention to a special kind of
deterministic probability kernel: There is a function
$\phi_{\y}(\sigma)$ that depends only on $\sigma$ through $\y^o$,
and $T_{\y}(\sigma, \sigma'_{\y})=2\delta(\phi_{\y}(\sigma),
\sigma'_{\y})$.

Our basic assumption is that the original interaction $J$ lies in
a Banach space $\mathcal{B}_r$, with norm
\begin{equation}
||J||_r=\sup_{\x \in \L}\sum_{X: \x \in X}|J(X)|e^{rl(\x, X)},
\end{equation}
where the constant $r\geq 0$, $d$ is a metric on $\L$, and $l(\x,
X)=\sup\{d(\x,\y): \y \in X\}$, with the convention that
$l(\x,\emptyset)=0$. Correspondingly, there is a paired Banach
space $\mathcal{B}_r^*$. As
\begin{eqnarray*}
|\sum_X J_1(X)J_2(X)|& \leq &\sum_X |J_1(X)|\sum_{\x \in
X}\frac{1}{|X|}|J_2(X)|\\&=&\sum_{\x \in \L}\sum_{X: \x \in
X}\frac{1}{|X|}|J_1(X)||J_2(X)|\\& \leq &\sum_{\x \in \L}\sup_{\x
\in X}\frac{1}{|X|}|J_2(X)|e^{-rl(\x, X)}\sum_{X: \x \in
X}|J_1(X)|e^{rl(\x,X)}\\& \leq &\sup_{\x \in \L}\sum_{X: \x \in
X}|J_1(X)|e^{rl(\x,X)} \cdot \sum_{\x \in \L}\sup_{\x \in
X}\frac{1}{|X|}|J_2(X)|e^{-rl(\x,X)},
\end{eqnarray*}
a suitable $\mathcal{B}_r^*$ norm is defined by
\begin{equation}
||J||_r^*=\sum_{\x \in \L}\sup_{\x \in
X}\frac{1}{|X|}|J(X)|e^{-rl(\x, X)}.
\end{equation}
Notice that here $\mathcal{B}_r$ is technically not the dual space
of $\mathcal{B}_r^*$, and $\mathcal{B}_r^*$ is technically not the
dual space of $\mathcal{B}_r$. The spaces $\mathcal{B}_r$ and
$\mathcal{B}_r^*$ are paired in the sense that each one is part of
the dual space of the other, or in other words, each one consists
of continuous linear functions defined on the other. We study the
situation when $||J||_r=0$ (indication of infinite temperature).
We consider the spectrum of the linearization $\mathrm{L}(J)$ of
two commonly used RG transformations, decimation and deterministic
majority rule with odd blocking factor. We show that this spectrum
is of an unusual kind: dense point spectrum for which the adjoint
operators $\mathrm{L}^*(J)$ have no point spectrum at all, but
only residual spectrum.

\begin{remark}
In this paper, spectrum is crudely divided into $3$ types
\cite{Dunford}: For a bounded linear operator $\A$ acting on a
Banach space $\A: B \rightarrow B$,
\begin{enumerate}
\item $\lambda$ is in the point spectrum $\iff$ there exists $B\ni
u\neq 0$, such that $(\A-\lambda)u=0$, i.e.,
$\mathrm{Kernel}(\A-\lambda \I)$ is nontrivial. \item $\lambda$ is
in the residual spectrum $\iff$ $\lambda$ is not in the point
spectrum, and $\overline{\mathrm{Range}(\A-\lambda \I)} \neq B$.
\item $\lambda$ is in the continuous spectrum $\iff$ $\lambda$ is
not in the point spectrum or the residual spectrum,
$\mathrm{Range}(\A-\lambda \I) \neq B$, and
$\overline{\mathrm{Range}(\A-\lambda \I)}=B$.
\end{enumerate}
This definition is too simple to fully capture the notion of
continuous spectrum, but it will be adequate for our purposes.
\end{remark}

Israel \cite{Israel} found the operator bound of $\mathrm{L}(J)$
for decimation in a Banach algebra setting, but did not go into
detail about the spectral type of this transformation. He also
examined the operator bound of $\mathrm{L}(J)$ for majority rule
on the triangular lattice. These results are extended by the
present investigation, which includes the spectral type of
$\mathrm{L}(J)$ and $\mathrm{L}^*(J)$ for decimation (Theorems
\ref{dec} and \ref{addec}) and majority rule (Theorems \ref{maj}
and \ref{admaj}). Even though this investigation is focused on the
RG transformation acting on a system very close to a trivial
interaction, it serves as a test case---after all, if it is
reasonably difficult to compute the spectrum of the RG map, then
one can get an idea of what to expect by computing in a simple
case. If even this case has bizarre spectral properties, then it
may serve as a lesson in what to expect in more general
situations.

\section{Some general results}
\label{general}
\begin{proposition}
The renormalized coupling constants $J'$ are given by the
expression
\begin{equation}
\label{JZ} J'(Z)=\sum_{\sigma'}\sigma'_Z \log(W(\sigma')),
\end{equation}
where $W(\sigma')$ is the frozen block spin partition function
given by
\begin{equation}
W(\sigma')=\sum_{\sigma}\prod_{\y \in
\L'}T_{\y}(\sigma,\sigma'_{\y})e^{\sum_X J(X)\sigma_X}.
\end{equation}
\end{proposition}

\begin{proof}
In order to write down an explicit expression of $J'$, we use
Fourier series on the group $\{+1,-1\}^{\L'}$. If
$H'(\sigma')=-\sum_Y J'(Y) \sigma'_Y$, then $J'(Z)=\sum_{\sigma'}
-H'(\sigma')\sigma'_Z$. We see that
\begin{multline}
J'(Z)=\sum_{\sigma'}\sigma'_Z\log\left(\sum_{\sigma}\prod_{\y \in
\L'}T_{\y}(\sigma,\sigma'_{\y})e^{\sum_X
J(X)\sigma_X}\right)\\+\sum_{\sigma'}\sigma'_Z\log\left(\sum_{\sigma'}e^{\sum_Y
J'(Y)\sigma'_Y}\right)-\sum_{\sigma'}\sigma'_Z\log\left(\sum_{\sigma}e^{\sum_X
J(X)\sigma_X}\right).
\end{multline}
An important observation here is that
$\log\left(\sum_{\sigma'}e^{\sum_Y J'(Y)\sigma'_Y}\right)$ and
$\log\left(\sum_{\sigma}e^{\sum_X J(X)\sigma_X}\right)$ are
constants with respect to $\sigma'_Z$; thus, when summing over all
possible image configurations $\sigma'$, they both vanish.
\end{proof}

\begin{proposition}
Suppose the original interaction $J$ is at infinite temperature.
Then for every subset $W$ of the original lattice and every subset
$Z$ of the image lattice, the partial derivative $\frac{\partial
J'(Z)}{\partial J(W)}$ of the RG transformation is given by the
expression
\begin{equation}
\label{partial} \frac{\partial J'(Z)}{\partial
J(W)}=\sum_{\sigma}\sum_{\sigma'}\prod_{\y \in
\L'}T_{\y}(\sigma,\sigma'_{\y})\sigma_W\sigma'_Z.
\end{equation}
\end{proposition}

\begin{proof}
We take the derivative of both sides of (\ref{JZ}) with respect to
$J(W)$.
\begin{equation}
\frac{\partial J'(Z)}{\partial J(W)}=\sum_{\sigma'}\sigma'_Z
\frac{\sum_{\sigma}\prod_{\y \in
\L'}T_{\y}(\sigma,\sigma'_{\y})e^{\sum_X
J(X)\sigma_X}\sigma_W}{\sum_{\sigma}\prod_{\y \in
\L'}T_{\y}(\sigma,\sigma'_{\y})e^{\sum_X J(X)\sigma_X}}.
\end{equation}
When $J$ is at infinite temperature, i.e., $||J||_r=0$, $J(X)=0$
for every subset $X$ of the original lattice.
\end{proof}

\begin{definition}
For every subset $Z$ of the image lattice, the linearization
$\mathrm{L}(J)$ of the RG transformation for $J$ at infinite
temperature is given by a linear function of $K$ (which indicates
variation from infinite temperature),
\begin{equation}
\mathrm{L}(J)K(Z)=\sum_W \frac{\partial J'(Z)}{\partial J(W)}
K(W),
\end{equation}
where $W$ ranges over all finite subsets of the original lattice.
\end{definition}

\begin{definition}
The adjoint of the linearization $\mathrm{L}^*(J)$ of the RG
transformation for $J$ at infinite temperature is characterized by
the usual correspondence between adjoint operators,
\begin{equation}
\label{corresp} \sum_X K_1(X)\mathrm{L}(J)K_2(X)=\sum_Y
K_2(Y)\mathrm{L}^*(J)K_1(Y),
\end{equation}
where $X$ ranges over all finite subsets of the image lattice, and
$Y$ ranges over all finite subsets of the original lattice.
\end{definition}

\begin{definition}
A constant pure magnetic field is one such that $K(X) = 0$ except
for one-point sets $\{\x\}$, where $K(\{\x\}) = m$, a constant.
\end{definition}

\section{Spectrum of the linearization of decimation transformation and its adjoint at infinite temperature}
\label{decimation}
\begin{proposition}
Consider decimation transformation with blocking factor $b$ and a
probability kernel defined by
\begin{eqnarray}
\phi_{\y}(\sigma)=\sigma_{b\y},
\end{eqnarray}
where $b\y=b(y_1,...,y_d)=(by_1,...,by_d)$. Suppose the original
interaction $J$ is at infinite temperature. Then for every subset
$Z$ of the image lattice, the linearization $\mathrm{L}(J)$ of
this transformation is given by the expression
\begin{equation}
\label{decJ} \mathrm{L}(J) K(Z)=K(b Z),
\end{equation}
where $b Z=\cup_{\z\in Z}\{b \z\}$.
\end{proposition}

\begin{proof}
We evaluate (\ref{partial}) explicitly:
\begin{eqnarray}
\frac{\partial J'(Z)}{\partial J(W)}=
\sum_{\sigma} \delta(W, b Z)=\delta(W, b Z),
\end{eqnarray}
where $\delta$ is the Kronecker delta function.
\end{proof}

\begin{proposition}
Consider the adjoint of decimation transformation with blocking
factor $b$ and a probability kernel defined by
\begin{eqnarray}
\phi_{\y}(\sigma)=\sigma_{b\y}.
\end{eqnarray}
Suppose the original interaction $J$ is at infinite temperature.
Then for every subset $Z$ of the original lattice, the adjoint of
the linearization $\mathrm{L}^*(J)$ of this transformation is
given by the expression
\begin{eqnarray}
\label{addecJ} \mathrm{L}^*(J)K(Z)=\left\{\begin{array}{ll}
K(Y) & \mbox{if $Z=bY$};\\
0 & \mbox{otherwise}.\end{array} \right.
\end{eqnarray}
\end{proposition}

\begin{proof}
We notice that in this case, (\ref{corresp}) becomes
\begin{eqnarray}
\sum_X K_1(X)\mathrm{L}(J)K_2(X)&=&\sum_X K_1(X)K_2(b X).
\end{eqnarray}
Without loss of generality, we assume
$\mathrm{L}^*(J)K(\{\mathbf{0}\})=0$, which amounts to an index
shift.
\end{proof}

\begin{theorem}[Israel]
\label{dec} Suppose the original interaction $J$ is at infinite
temperature. Then in the Banach Space $\mathcal{B}_r$, the
spectrum of the linearization of the decimation transformation
$\mathrm{L}(J)$ is all point spectrum, $|\lambda| \leq 1$.
\end{theorem}

\begin{proof}
The proof of this theorem follows from several propositions.
\end{proof}

\begin{proposition}
$||\mathrm{L}(J)||=1$.
\end{proposition}

\begin{proof}
We check that for each fixed $\x \in \L$, $\sum_{X: \x \in
X}|\mathrm{L}(J) K(X)|e^{rl(\x, X)}\leq ||K||_r$, which would
imply $||\mathrm{L}(J)||\leq 1$. By (\ref{decJ}),
\begin{equation*}
\sum_{X: \x \in X}|\mathrm{L}(J) K(X)|e^{rl(\x, X)}=\sum_{X: \x
\in X}|K(b X)|e^{rl(\x, X)}
\end{equation*}
\begin{equation}
\leq\sum_{X: b \x \in b X}|K(b X)|e^{rl(b \x, b X)}\leq \sum_{X: b
\x \in X}|K(X)|e^{rl(b \x, X)}\leq ||K||_r.
\end{equation}
The claim is verified when we realize that
a constant pure magnetic field is an eigenvector with eigenvalue
$1$.
\end{proof}

\begin{corollary}
Every eigenvalue $\lambda$ of $\mathrm{L}(J)$ satisfies
$|\lambda|\leq 1$.
\end{corollary}

\begin{proposition}
Every $|\lambda| \leq 1$ is an eigenvalue.
\end{proposition}

\begin{proof}
For a generic $\lambda$, we display one eigenvector here. In fact,
with some further thought, it is not hard to show that there are
infinitely many eigenvectors for each $\lambda$. The eigenvector
$K$ is defined by
\begin{equation}
K(\{(b^n,0,...,0)\})=\lambda^n K(\{(1,0,...,0)\})=\lambda^n
\end{equation}
for $n\geq 0$, and for all the other subsets $X$, $K(X)$ is set to
zero.
\end{proof}

Moreover, we have stricter restrictions on the eigenvector $K$
that lies in a Banach space $\mathcal{B}_r: r>0$.

\begin{proposition}
\label{r} For $\lambda \neq 0$ and for every finite subset
$|X|>1$, we must have $K(X)=0$ for the eigenvector $K$.
\end{proposition}

\begin{proof}
This follows from the observation that we can always pick a site,
say $\x$, in $X$, such that $l(\x,X)>0$. As a result, $b^n \x$ is
a site in $b^n X$, and $l(b^{n} \x, b^{n} X)=b^{n}l(\x, X)>0$.
Since $\mathrm{L}(J) K(X)=K(b X)$, we must have $K(b^n
X)=\lambda^n K(X)$. Then due to the fact that $K$ is an
eigenvector, we need to ensure that
\begin{equation}
|\lambda|^n |K(X)|e^{rb^n l(\x, X)}<\infty.
\end{equation}
\end{proof}

The following statements concern translation-invariant
Hamiltonians. In this case, it is believed that the RG map should
be almost a contraction near the trivial fixed point. Almost means
that except for a few degrees of freedom (maybe just one) it
should be a contraction. If one restricts oneself to even
interactions, then it should actually be a
contraction---reflecting the fact that if we start with an even
interaction in the very high temperature phase, then the RG map
would drive it to the zero interaction.

\begin{proposition}
\label{tran} Restricted to the translation-invariant even subspace
of $\mathcal{B}_r: r=0$, the point spectrum of $\L$ is
$|\lambda|<1$.
\end{proposition}

\begin{proof}
For $|\lambda|<1$, the eigenvector $K$ may be defined by
\begin{equation}
K(\{\x,\y\})=1
\end{equation}
for $\x,\y$ that are nearest-neighbors, and in general,
$K(bX)=\lambda K(X)$. However, no such eigenvector would work for
$|\lambda|=1$. Suppose the nontrivial eigenvector $K(X)=m\neq 0$
for some finite subset $X:|X|>1$. Due to translation-invariance,
for arbitrary $n$, all sets $Y$ with the same shape as $b^nX$ will
have $|K(Y)|=m$. In particular, there will be infinitely many
subsets $Z$ containing $\mathbf{0}$ with $|K(Z)|=m$, which implies
$||K||_r=\infty$.
\end{proof}

\begin{proposition}
Restricted to the translation-invariant even subspace of
$\mathcal{B}_r: r>0$, the point spectrum of $\L$ is $\lambda=0$.
\end{proposition}

\begin{proof}
This follows from Propositions \ref{r} and \ref{tran}.
\end{proof}

\begin{theorem}
\label{addec} Suppose the original interaction $J$ is at infinite
temperature. Then in the Banach Space $\mathcal{B}_r^*$, the
spectrum of the adjoint of the linearization of the decimation
transformation $\mathrm{L}^*(J)$ is all residual spectrum,
$|\lambda|\leq 1$.
\end{theorem}

\begin{proof}
The proof of this theorem follows from several propositions.
\end{proof}

\begin{proposition}
$||\mathrm{L}^*(J)|| \leq 1$.
\end{proposition}

\begin{proof}
By (\ref{addecJ}),
\begin{eqnarray}
\sum_{\x \in \L}\sup_{\x \in X}\frac{1}{|X|}|\mathrm{L}^*(J)
K(X)|e^{-rl(\x, X)}&\leq& \sum_{b\x \in \L}\sup_{b\x \in bX}
\frac{1}{|X|}|K(X)|e^{-rl(b\x, bX)}\notag\\&\leq& \sum_{\x \in
\L}\sup_{\x \in X}\frac{1}{|X|}|K(X)|e^{-rl(\x, X)}.
\end{eqnarray}
\end{proof}

\begin{proposition}
For every $\lambda \neq 0$, there is no nontrivial eigenvector.
\end{proposition}

\begin{proof}
Fix an arbitrary finite subset $X$ of the infinite lattice, after
a finite number of iterations of $\mathrm{L}^*(J)$ (say $n$
times), $X$ will not be of the form $bY$ for some $Y$. Thus
$\lambda^{n+1} K(X)=(\mathrm{L}^*(J))^{n+1}K(X)=0$, which implies
$K(X)=0$.
\end{proof}

\begin{proposition}
For $\lambda=0$, there is no nontrivial eigenvector.
\end{proposition}

\begin{proof}
Suppose the nontrivial eigenvector $K(X)=m \neq 0$ for some finite
subset $X$, then the crucial fact that we can always find $Y$,
with $\mathrm{L}^*(J)K(Y)=K(X)$ will do the job. As
$\mathrm{L}^*(J)K(Y)=\lambda K(Y)=0$, we reach a contradiction.
\end{proof}

\begin{corollary}
In the Banach Space $\mathcal{B}_r^*$, the point spectrum of
$\mathrm{L}^*(J)$ is empty.
\end{corollary}

\noindent \textit{Proof of Theorem \ref{addec} continued.} The
only thing left to show now is that
$\overline{\mathrm{Range}(\lambda I-\mathrm{L}^*(J))}\neq
\mathcal{B}_r^*$ for $|\lambda|\leq 1$. Define
$K(\{(1,0,...,0)\})=1$, and $K(X)=0$ for all other subsets $X$. We
will show that $K$ can not be approximated by any $K'$ in
$\mathrm{Range}(\lambda I-\mathrm{L}^*(J))$ within distance $1/2$.
To see this, note that for $n\geq 0$,
\begin{eqnarray}
K'(\{(b^{n+1},0,...,0)\})=\lambda
S(\{(b^{n+1},0,...,0)\})-S(\{(b^n,0,...,0)\})
\end{eqnarray}
for some $S$ that lies in $\mathcal{B}_r^*$. Suppose
\begin{eqnarray}
\frac{1}{2}&\geq&||K-K'||_r^*=\sum_{\x \in \L}\sup_{\x \in X}\frac{1}{|X|}|K(X)-K'(X)|e^{-rl(\x, X)}\notag \\
&\geq&\sum_{\x=(b^n,0,...,0)}\sup_{\x \in
X}\frac{1}{|X|}|K(X)-K'(X)|e^{-rl(\x, X)}\notag \\ &\geq&
\sum_{n=0}^{\infty}|K(\{(b^n,0,...,0)\})-K'(\{(b^n,0,...,0)\})|
\end{eqnarray}
\begin{equation}
=|\lambda S(\{(1,0,...,0)\})-1|+|\lambda
S(\{(b,0,...,0)\})-S(\{(1,0,...,0)\})|+\cdots
\end{equation}
Then, as $|\lambda|\leq 1$, for any $n\geq 0$,
\begin{multline}
\frac{1}{2}\geq|\lambda^{n+1} S(\{(b^n,0,...,0)\})-\lambda^n
S(\{(b^{n-1},0,...,0)\})|+\cdots\\+|\lambda^2
S(\{(b,0,...,0)\})-\lambda S(\{(1,0,...,0)\})|+|\lambda
S(\{(1,0,...,0)\})-1|.
\end{multline}
By the triangle inequality, this implies
\begin{equation}
|\lambda^{n+1} S(\{(b^n,0,...,0)\})-1|\leq \frac{1}{2},
\end{equation}
which further implies
\begin{equation}
|\lambda^{n+1}S(\{(b^n,0,...,0)\})|\geq \frac{1}{2}.
\end{equation}
Using $|\lambda|\leq 1$ again, we have
\begin{equation}
|S(\{(b^n,0,...,0)\})|\geq \frac{1}{2}.
\end{equation}
But then,
\begin{equation}
||S||_r^*=\sum_{\x \in \L}\sup_{\x \in
X}\frac{1}{|X|}|S(X)|e^{-rl(\x, X)}\geq
\sum_{n=0}^{\infty}|S(\{(b^n,0,...,0)\})|=\infty.
\end{equation}

\begin{remark}
Notice the similarity between the adjoint operators
$\mathrm{L}(J)$/$\mathrm{L}^*(J)$ in our Banach spaces and
left/right translation in $l^{\infty}$/$l^1$. $\mathrm{L}(J)$ acts
like left translation and $\mathrm{L}^*(J)$ acts like right
translation on sequences $(X,b X,...)$ for all possible subsets
$X$. Moreover, ignoring multiplicity of the eigenvalues, the
spectrum of $\mathrm{L}(J)$ is the same as that of left
translation in $l^{\infty}$, and the spectrum of $\mathrm{L}^*(J)$
is the same as that of right translation in $l^1$. This might be
related to the fact that the norms in our Banach spaces are
something like combinations of $l^{\infty}$ and $l^1$ norms.
\end{remark}

\section{Spectrum of the linearization of majority rule transformation and its adjoint at infinite temperature}
\label{majority} For notational convenience, in this section, we
set $s=b^d$ and $\nu=\tbinom{s-1}{\frac{s-1}{2}}/2^{s-1}$.

\begin{proposition}
Consider majority rule transformation with odd blocking factor $b$
and a probability kernel defined by
\begin{eqnarray}
\phi_{\y}(\sigma)=\sign\left(\sum_{\x\in \y^o}\sigma_{\x}\right).
\end{eqnarray}
Suppose the original interaction $J$ is at infinite temperature.
Then for every subset $Z$ of the image lattice, the linearization
$\mathrm{L}(J)$ of this transformation is given by the expression
\begin{equation}
\mathrm{L}(J) K(Z)=\sum_{W: W \subset Z^o}\prod_{\z \in Z}\chi(W
\cap \z^o)K(W),
\end{equation}
where $\chi(W\cap
\z^o)=\sum_{\sigma}\sum_{\sigma'}T_{\z}(\sigma,\sigma'_{\z})\sigma_{W\cap
\z^o}\sigma'_\z$.
\end{proposition}

\begin{proof}
We evaluate (\ref{partial}) explicitly:
\begin{equation}
\frac{\partial J'(Z)}{\partial J(W)}=\sum_{\sigma}\sigma_{W
\texttt{\char92} Z^o}\prod_{\z\notin
Z}\sum_{\sigma'}T_{\z}(\sigma,\sigma'_{\z})\prod_{\z\in
Z}\sum_{\sigma}\sum_{\sigma'}T_{\z}(\sigma,\sigma'_{\z})\sigma_{W\cap
\z^o}\sigma'_{\z}.
\end{equation}
Since $\sum_{\sigma}\sigma_{W \texttt{\char92} Z^o}=0$ for $W$ not
completely contained inside $Z^o$, it follows that $W \subset
Z^o$.
\end{proof}

\begin{proposition}
Consider majority rule transformation with odd blocking factor $b$
and a probability kernel defined by
\begin{eqnarray}
\phi_{\y}(\sigma)=\sign\left(\sum_{\x\in \y^o}\sigma_{\x}\right).
\end{eqnarray}
Suppose the original interaction $J$ is at infinite temperature.
Then for every subset $Z$ of the original lattice, the adjoint of
the linearization $\mathrm{L}^*(J)$ of this transformation is
given by the expression
\begin{eqnarray}
\label{admajJ}
\mathrm{L}^*(J)K(Z)=\prod_{W_{\n}}\chi(W_{\n})K(\cup \{\n\}),
\end{eqnarray}
where $Z=\cup W_{\n}$ and $W_{\n}\subset \n^o$.
\end{proposition}

\begin{proof}
We notice that in this case, (\ref{corresp}) becomes
\begin{eqnarray}
\sum_X K_1(X)\mathrm{L}(J)K_2(X)&=&\sum_X K_1(X)\sum_{Y: Y \subset
X^o}\prod_{\x \in X}\chi(Y \cap \x^o)K_2(Y)\notag
\\&=&\sum_{Y=\cup W_{\n}}K_2(Y)\prod_{W_{\n}}\chi(W_{\n})K_1(\cup \{\n\}).
\end{eqnarray}
\end{proof}

\begin{theorem}
\label{maj} Suppose the original interaction $J$ is at infinite
temperature. Then in the Banach Space $\mathcal{B}_r$, the
spectrum of the linearization of the majority rule transformation
$\mathrm{L}(J)$ is all point spectrum, $|\lambda|\leq s\nu$.
\end{theorem}

\begin{proof}
The proof of this theorem follows from several propositions.
\end{proof}

\begin{proposition}
Consider Ising-type spin system on an odd polygon $A$ with
cardinality $|A|$. Fix a certain vertex $V$ and a certain subset
$W$ of the vertices. If $\sigma'_a\in \{+1,-1\}$ satisfies
$\sigma_A\sigma'_a>0$, then
\begin{equation}
\label{polygon}
|\sum_{\sigma}\sigma_W\sigma'_a| \leq
\sum_{\sigma}\sigma_{V}\sigma'_a=\tbinom
  {|A|-1}{\frac{|A|-1}{2}}/2^{|A|-1},
\end{equation}
where $\tbinom nk$ is the binomial coefficient.
\end{proposition}

\begin{proof}
We first show that $\sum_{\sigma}\sigma_W\sigma'_a=0$ for any $W$
with even cardinality. This is due to a symmetry argument. If
there is a spin configuration with $\sigma_W\sigma'_a=1$, then
flipping the spins at every vertex, we will have a
configuration with
$\sigma_W\sigma'_a=(-1)^{|W|}(-1)=(-1)^{|W|+1}=-1$. Vice versa. Thus the total sum
will be zero.

Next we investigate into the special case
$\sum_{\sigma}\sigma_{V}\sigma'_a$ where $V$ is any fixed vertex.
The explicit calculation is easy to carry out. Due to symmetry, we
only consider $\sigma_V=1$ in the following, and there are $|A|-1$
vertices for which the spins are yet to be assigned.
\begin{enumerate}
    \item $\sigma'_a=1$, if there are more $1$'s than $-1$'s in the
overall spin configuration, i.e., as long as the number of $-1$'s
does not exceed $\frac{|A|-1}{2}$. It is not hard to see that
there are
$\tbinom{|A|-1}{0}+\tbinom{|A|-1}{1}+\cdots+\tbinom{|A|-1}{\frac{|A|-1}{2}}$
of them.
    \item $\sigma'_a=-1$, if there are more $-1$'s than $1$'s in the overall
spin configuration, i.e., as long as the number of $-1$'s exceeds
$\frac{|A|-1}{2}$. Again, it is not hard to see that there are
$\tbinom{|A|-1}{\frac{|A|+1}{2}}+
\tbinom{|A|-1}{\frac{|A|+3}{2}}+\cdots+\tbinom{|A|-1}{|A|-1}=
\tbinom{|A|-1}{\frac{|A|-3}{2}}+\tbinom{|A|-1}{\frac{|A|-5}{2}}+\cdots+\tbinom{|A|-1}{0}$
of them.
\end{enumerate}
In conclusion, when $\sigma_V=1$, there are
$\tbinom{|A|-1}{\frac{|A|-1}{2}}$ more spin configurations for
$\sigma'_a$ to be $1$ rather than to be $-1$. Similar result holds
for $\sigma_V=-1$. Thus considering all possible spin
configurations, there are $2\tbinom{|A|-1}{\frac{|A|-1}{2}}$ more
spin configurations for $\sigma_V\sigma'_a$ to be $1$ rather than
to be $-1$. It follows that
$\sum_{\sigma}\sigma_{V}\sigma'_a=\tbinom{|A|-1}{\frac{|A|-1}{2}}/2^{|A|-1}$.

Finally we consider $\sum_{\sigma}\sigma_W\sigma'_a$ for any $W$
with odd cardinality. Without loss of generality, suppose
$V\subset W$. For a fixed spin configuration,
$\sigma_V\sigma'_a\neq \sigma_W\sigma'_a$ can only occur when
there is an odd number of $-1$'s and an odd number of $1$'s in the
spin configuration for vertices in $W \texttt{\char92} V$. For
such a configuration, we notice the following important fact:
Suppose it has the extra property that unequal numbers of $-1$'s
and $1$'s are assigned for the remaining $|A|-1$ vertices of $A
\texttt{\char92} V$, then if we flip the spins at every vertex
other than $V$, $\sigma_V\sigma'_a$ will change sign. Moreover, at
the same time, the sign of $\sigma_W\sigma'_a$ also changes, so
the total sum does not change. Therefore, we see that the
difference in $\sum_{\sigma}\sigma_W\sigma'_a$ and
$\sum_{\sigma}\sigma_V\sigma'_a$ can only be caused by the
following scenario: Equal numbers of $-1$'s and $1$'s are assigned
for the remaining $|A|-1$ vertices of $A \texttt{\char92} V$, and
there is an odd number of $-1$'s and an odd number of $1$'s in the
spin configuration for vertices in $W \texttt{\char92} V$. It is
not hard to see that there are at most
$2\tbinom{|A|-1}{\frac{|A|-1}{2}}$ of them. Thus
$\sum_{\sigma}\sigma_W\sigma'_a$ varies between
$-\tbinom{|A|-1}{\frac{|A|-1}{2}}/2^{|A|-1}$ and
$\tbinom{|A|-1}{\frac{|A|-1}{2}}/2^{|A|-1}$, and our claim
follows.
\end{proof}

\begin{proposition}
$||\mathrm{L}(J)||=s\nu$.
\end{proposition}

\begin{proof}
We check that for each fixed $\x \in \L$, $\sum_{X: \x \in
X}|\mathrm{L}(J) K(X)|e^{rl(\x,X)} \leq s\nu||K||_r$, which would
imply $||\mathrm{L}(J)||\leq s\nu$. As $\x \in X$, $\mathrm{L}(J)
K(X)$ is a linear combination of $K(Y)$'s, each one with
coefficient bounded above by $\nu$ by (\ref{polygon}). Ignoring
the coefficients of $K(Y)$'s, we can then collect terms according
to which one of the sites in $\x^o$ belongs to $Y$. (When $|Y\cap
\x^o|>1$, $K(Y)$ can be classified into either one of the $s$
groups.) Moreover, each $Y$ has size no smaller than $X$, the
exponential factor changes to a larger quantity after the action
of $\mathrm{L}(J)$. We see that each collection is bounded above
by $||K||_r$ by definition. The claim is verified when we realize
that a constant pure magnetic field is an eigenvector with
eigenvalue $s\nu$.
\end{proof}

\begin{corollary}
Every eigenvalue $|\lambda| \leq s\nu$.
\end{corollary}

\begin{proposition}
Every $|\lambda| \leq s\nu$ is an eigenvalue.
\end{proposition}

\begin{proof}
For a generic $\lambda$, we display one eigenvector here. In fact,
with some further thought, it is not hard to show that there are
infinitely many eigenvectors for each $\lambda$. The eigenvector
$K$ is defined by
\begin{equation}
K(\{(\frac{b-1}{2},...,\frac{b-1}{2})\})=\lambda/\nu-(s-1),
\end{equation}
and
\begin{equation}
K(\{\x\})=1
\end{equation}
for $(\frac{b-1}{2},...,\frac{b-1}{2})\neq \x\in \mathbf{0}^o$. In
general, for $\n \neq \mathbf{0}$, $K$ is defined by
$sK(\{\m\})=\lambda/\nu K(\{\n\})$ for $\m\in \n^o$. For all the
other subsets $X$, $K(X)$ is set to zero.
\end{proof}

\begin{corollary}
The spectrum of $\mathrm{L}(J)$ diverges as
$\sqrt{\frac{2s}{\pi}}$ as the blocking factor $b$ gets large.
\end{corollary}

\begin{proof}
This follows from an easy application of Stirling's formula:
\begin{eqnarray}
s\nu\sim\frac{s \sqrt{2\pi\cdot
(s-1)}(s-1)^{s-1}e^{-(s-1)}}{2\pi\frac{s-1}{2}\left(\frac{s-1}{2}\right)^{s-1}e^{-(s-1)}2^{s-1}}\sim\sqrt{\frac{2s}{\pi}}.
\end{eqnarray}
\end{proof}

\begin{theorem}
\label{admaj} Suppose the original interaction $J$ is at infinite
temperature. Then in the Banach Space $\mathcal{B}_r^*$, the point
spectrum of the adjoint of the linearization of the majority rule
transformation $\mathrm{L}^*(J)$ is empty. Moreover, every
$|\lambda|\leq \nu$ is in the residual spectrum of
$\mathrm{L}^*(J)$.
\end{theorem}

\begin{proof}
The proof of this theorem follows from several propositions.
\end{proof}

\begin{proposition}
For every $\lambda \neq 0$ and $\lambda \neq \nu$, there is no
nontrivial eigenvector.
\end{proposition}

\begin{proof}
Fix an arbitrary finite subset $X$. For $\lambda \neq 0$, $K(X)$
is either zero or a nonzero constant multiple of
$K(\{\mathbf{0}\})$ as a result of the action of
$\mathrm{L}^*(J)$. In particular, $\lambda
K(\{\mathbf{0}\})=\mathrm{L}^*(J)K(\{\mathbf{0}\})=\nu
K(\{\mathbf{0}\})$, which implies that $K(\{\mathbf{0}\})=0$.
\end{proof}

\begin{proposition}
For $\lambda=0$, there is no nontrivial eigenvector.
\end{proposition}

\begin{proof}
Suppose the nontrivial eigenvector $K(X)=m \neq 0$ for some finite
subset $X$, then the crucial fact that we can always find $Y$,
with $\mathrm{L}^*(J)K(Y)$ a nonzero constant multiple of $K(X)$
will do the job. As $\mathrm{L}^*(J)K(Y)=\lambda K(Y)=0$, we reach
a contradiction.
\end{proof}

\begin{proposition}
For $\lambda=\nu$, every nontrivial eigenvector has norm infinity.
\end{proposition}

\begin{proof}
We must have $K(\{\mathbf{0}\})=m\neq 0$ in order for $K$ to be
nontrivial. As
\begin{equation}
\nu K(\{\x\})=\mathrm{L}^*(J)K(\{\x\})=\nu K(\{\mathbf{0}\})
\end{equation}
for $\x\in \mathbf{0}^o$, we see that $K(\{\x\})=m$ also.
Following similar fashion, $K(\{\n\})=m$ for arbitrary $\n$. But
then, $||K||_r^*=\infty$.
\end{proof}

\noindent \textit{Proof of Theorem \ref{admaj} continued.} The
only thing left to show now is that
$\overline{\mathrm{Range}(\lambda I-\mathrm{L}^*(J))}\neq
\mathcal{B}_r^*$ for $|\lambda|\leq \nu$. Define
$K(\{(0,0,...,0)\})=1$, and $K(X)=0$ for all other subsets $X$. We
will show that $K$ can not be approximated by any $K'$ in
$\mathrm{Range}(\lambda I-\mathrm{L}^*(J))$ within distance $1/4$.
To see this, note that for $n\geq 0$,
\begin{eqnarray}
K'(\{(b^{n+1},0,...,0)\})=\lambda S(\{(b^{n+1},0,...,0)\})-\nu
S(\{(b^n,0,...,0)\})
\end{eqnarray}
for some $S$ that lies in $\mathcal{B}_r^*$. And in particular,
\begin{equation}
K'(\{0,...,0\})=(\lambda-\nu)S(\{0,...,0\}).
\end{equation}
Suppose
\begin{equation*}
\frac{1}{4}\geq||K-K'||_r^*=\sum_{\x \in \L}\sup_{\x \in
X}\frac{1}{|X|}|K(X)-K'(X)|e^{-rl(\x, X)}
\end{equation*}
\begin{equation*}
\geq |K(\{(0,0,...,0)\})-K'(\{(0,0,...,0)\})|+
\sum_{n=0}^{\infty}|K(\{(b^n,0,...,0)\})-K'(\{(b^n,0,...,0)\})|
\end{equation*}
\begin{equation}
=|(\lambda-\nu) S(\{(0,0,...,0)\})-1|+|\lambda
S(\{(1,0,...,0)\})-\nu S(\{(0,0,...,0)\})|+\cdots
\end{equation}
Then, as $|\lambda|\leq \nu\leq \frac{1}{2}$, for any $n\geq 0$,
\begin{equation*}
\frac{1}{2}\geq|(\frac{\lambda}{\nu})^{n+1}(\lambda-\nu)
S(\{(b^n,0,...,0)\})-(\frac{\lambda}{\nu})^{n}(\lambda-\nu)
S(\{(b^{n-1},0,...,0)\})|+\cdots
\end{equation*}
\begin{equation}
+|\frac{\lambda}{\nu}(\lambda-\nu)
S(\{(1,0,...,0)\})-(\lambda-\nu)
S(\{(0,0,...,0)\})|+|(\lambda-\nu) S(\{(0,0,...,0)\})-1|.
\end{equation}
By the triangle inequality, this implies
\begin{equation}
|(\frac{\lambda}{\nu})^{n+1}(\lambda-\nu)
S(\{(b^n,0,...,0)\})-1|\leq \frac{1}{2},
\end{equation}
which further implies
\begin{equation}
|(\frac{\lambda}{\nu})^{n+1}(\lambda-\nu)S(\{(b^n,0,...,0)\})|\geq
\frac{1}{2}.
\end{equation}
Using $|\lambda|\leq \nu\leq \frac{1}{2}$ again, we have
\begin{equation}
|S(\{(b^n,0,...,0)\})|\geq \frac{1}{2}.
\end{equation}
But then,
\begin{equation}
||S||_r^*=\sum_{\x \in \L}\sup_{\x \in
X}\frac{1}{|X|}|S(X)|e^{-rl(\x, X)}\geq
\sum_{n=0}^{\infty}|S(\{(b^n,0,...,0)\})|=\infty.
\end{equation}

\begin{acknowledgement}
This work began at the Isaac Newton Institute in Cambridge during
the $2008$ program in Combinatorics and Statistical Mechanics,
organized by Alan Sokal. The author owes deep gratitude to her
advisor Bill Faris for his continued help and support. She also
thanks Tom Kennedy, Doug Pickrell, and Bob Sims for their kind and
helpful suggestions and comments. The author appreciated the
opportunity to talk about this work in the $2009$ workshop in
Renormalization Group and Statistical Mechanics in Vancouver,
organized by David Brydges, Joel Feldman, and A.C.D. van Enter,
and is grateful for the feedback from many of these people.
\end{acknowledgement}


\begin{thebibliography}{2}
\bibitem{Dunford} Dunford, N., Schwartz, J.: Linear Operators, Part I. Interscience Publishers, New York (1958)

\bibitem{Israel} Israel, R.B.: Banach algebras and Kadanoff
transformations. In: Fritz, J., Lebowitz, J.L., Sz\'{a}sz, D.
(eds.) Random Fields, Vol. II, pp. 593-608. North-Holland,
Amsterdam (1981)
\end{thebibliography}
\end{document}